\documentclass[10pt, twocolumn]{IEEEtran}

\usepackage{graphicx}  
\usepackage{amsmath, amsopn, psfrag, amsthm} 
\usepackage{amssymb}
\usepackage{color}
\usepackage{psfrag}
\usepackage{epsfig}
\usepackage{breakcites}
\usepackage{amsthm}
\usepackage{lipsum}

\usepackage{cite}
\usepackage{amsmath,amssymb,amsfonts,amsthm}
\usepackage{algorithmic}
\usepackage{graphicx}
\usepackage{textcomp}
\usepackage{cite} 
\graphicspath{{./img/}}
\DeclareGraphicsExtensions{.pdf,.jpeg,.png}
\newtheorem{lemma}{Lemma}
\newtheorem{remark}{Remark}

\newtheorem{corollary}{Corollary}
\def\BibTeX{{\rm B\kern-.05em{\sc i\kern-.025em b}\kern-.08em
		T\kern-.1667em\lower.7ex\hbox{E}\kern-.125emX}}
\usepackage{amssymb,amsmath,mathtools}

\begin{document}

\title{ Alternative Formulations for the Fluctuating Two-Ray Fading Model 
}

\author{Maryam~Olyaee,~Juan M. Romero-Jerez,~ F. Javier Lopez-Martinez and Andrea J. Goldsmith

\thanks{
This work has been
funded in part by the Spanish Government and the European Fund for
Regional Development FEDER (project TEC2017-87913-R) and by Junta de
Andalucia (project P18-RT-3175).}
\thanks{
M. Olyaee, J. M. Romero-Jerez, and F. J. L\'opez-Mart\'inez are with Communications and Signal Processing Lab, Instituto Universitario de Investigaci\'on en Telecomunicaci\'on (TELMA), Universidad de M\'alaga, CEI Andaluc\'ia TECH.
ETSI Telecomunicaci\'on, Bulevar Louis Pasteur 35, 29010 M\'alaga, Spain. (e-mails: maryam.olyaee@yahoo.com, romero@dte.uma.es, fjlopezm@ic.uma.es).}
\thanks{
A. Goldsmith is with the Electrical and Computer
Engineering Department, Princeton University, Princeton, NJ, 08544 USA,
(e-mail: goldsmith@princeton.edu).}
\thanks{This work has been submitted to the IEEE for publication. Copyright may
be transferred without notice, after which this version may no longer be
accessible}}
	{}
	
\maketitle
\begin{abstract}
We present two alternative formulations for the distribution of the fluctuating two-ray (FTR) fading model, which simplify its 
statistical characterization and subsequent use for performance evaluation. New expressions for the probability density 
function (PDF) and cumulative distribution function of the FTR model are obtained based on the observation that the FTR fading 
distribution is described, for arbitrary $m$, as an underlying Rician Shadowed (RS) distribution with continuously varying 
parameter $K$, while for the special case of $m$ being an integer, the FTR fading model is described in terms of a finite 
number of underlying squared Nakagami-$m$ distributions. It is shown that the chief statistics and any performance metric that 
are computed by averaging over the PDF of the FTR fading model can be expressed in terms of a finite-range integral over the 
corresponding statistic or performance metric for the RS (for arbitrary $m$) or the Nakagami-$m$ (for integer $m$) fading 
models,  which have a simpler analytical characterization than the FTR model and for which many results are available in closed-
form. New expressions for some Laplace-domain statistics of interest are also obtained; these are used to exemplify the 
practical relevance of this new formulation for performance analysis.

\end{abstract}
\section{Introduction}
Stochastic channel propagation models are indispensable for the design, performance evaluation, and comparison of wireless communication systems. 
With developing communication technologies, existing channel models need to be improved 
as the new spectral bands, environments and use cases present propagation features that are not properly captured by 
the well-known classical models.
In this regard, the Fluctuating Two-Ray (FTR) fading model was introduced in \cite{ftr} as a general model to capture the rapid fluctuations of
the radio channel in a wide variety of environments, through only three physically-motivated fading model parameters: $K$, $\Delta$ and $m$.

The FTR fading model includes traditional fading models such as Rayleigh, Rice, Nakagami-$m$ and Hoyt as special cases, as well 
as other widely-used models such as Durgin's two-wave with diffuse power (TWDP) fading model. Moreover, the FTR model provides a better
match to field measurements than most existing stochastic fading models in 
different environments, in particular at 28 GHz \cite{h18}. Consequently, the FTR fading model has been widely used for the 
performance analysis of 
millimeter wave (mmWave) communications  ~\cite{mine1,mine2,PHY,DF,hadi1,hadi3}. Nevertheless, the generality of the model 
permits its
application to different environments, such as maritime communications \cite{maritime}; also, a wide variety of wireless 
scenarios and 
metrics have been been investigated considering the FTR fading model \cite{UAV,hadi2,AP1,PHY,ED,sumF,newR,ftr2021,correct,moment
}.
In \cite{mine1}, maximum ratio combining (MRC) was studied for a multi-antenna system considering the $m$ parameter to be an 
integer.
In \cite{mine2}, the coverage probability, average rate and bit error probability of a single-antenna downlink cellular system 
with inter-cell
interference under FTR fading was analyzed, again for integer $m$.
%
Performance analysis for the sum of squared FTR random variables was presented in \cite{sumF} and the performance 
of the equal gain combining (EGC) technique was studied in \cite{hadi2} for arbitrary $m$ in terms of infinite 
summations. Channel capacity with different power adaption methods under FTR fading was explored, also in terms of  infinite 
summations,  in \cite{AP1}.
In \cite{PHY} and \cite{ftr2021}  the performance of physical layer security in FTR fading channels was analyzed,
 relay network performance for the FTR fading was investigated in \cite{DF,hadi1,hadi3, UAV}; and \cite{ED} studied 
the performance of energy detection in FTR fading.

The analytical results in all the aforementioned works on FTR are given either for the case when the $m$ parameter of the model 
is an integer or in terms of infinite series when arbitrary $m$ is considered. This is due to the fact that, until now, two different 
expressions of the probability density function (PDF)
of the signal-to-noise ratio (SNR) in FTR fading are available in the literature.
The original formulation for the PDF of the SNR of this model was introduced in \cite{ftr}, and an analytical expression based 
on the hypergeometric function $\Phi_2$ \cite{Matlabprog} was given for the case of integer $m$. For the case of arbitrary $m$, 
the PDF in \cite{ftr} can be computed by a numerical inverse Laplace transformation over the moment generating function (MGF), 
which is provided in closed-form. 
Later, an alternative expression for the PDF of the SNR valid for arbitrary $m$ was proposed in \cite{newR,correct} based on 
the use of infinite series, and therefore all the performance results based on this formulation entail evaluating an infinite series, which
must be truncated to be  numerically computed, yielding error bounds, and empirical results are necessary to establish the rate of  
convergence (which needs to be demonstrated) and the tightness of the approximation, which in general will depend on the channel parameters values.

In this paper we introduce two alternative formulations for the PDF of the FTR model (one for arbitrary real $m$ and another 
for integer $m$) that avoid the use of infinite series as well as inverse Laplace transformations. These formulations are based on the 
observation that the FTR fading model can be expressed in terms of some underlying analytically simpler fading models like the Rician 
Shadowed (RS) \cite{Abdi2003} or Nakagami-$m$ models. Both newly proposed formulations provide a more 
complete statistical characterization of the model, thus opening the door to the performance analysis of different wireless 
scenarios and metrics
for the FTR model not found in the literature. 
In particular, the main contributions of this paper are:
\begin{itemize}
	\item 
It is demonstrated that, for arbitrary real $m$, any given metric or statistical function for FTR fading can be readily obtained
 by computing a finite-range integral whose integrand 
is the metric or function available for RS fading.
	\item 
When $m$ is an integer, it is demonstrated that any given metric or statistical function for FTR fading can be obtained
 by computing a finite-range integral whose integrand 
is the metric or function available for Nakagami-$m$ fading.
	\item 
New important statistics relevant to communication theory are obtained by employing these new frameworks: The generalized 
MGF (GMGF), the incomplete MGF (IMGF), for arbitrary $m$; and the incomplete generalized MGF (IGMGF), for integer $m$.
\end{itemize}

As the integrands in the presented formulations are continuous bounded functions (for a 
given average SNR), these finite-range integrals can be  
efficiently computed. Moreover, in several important cases the resulting integral can be solved in closed-form,
as is the case for the GMGF for arbitrary $m$.
Note that both the RS and Nagakami-$m$ fading models have been intensively investigated, particularly the latter. Therefore, 
all the 
available results for these fading models directly yield results for FTR fading in a straightforward manner, thus avoiding 
starting the required statistical analysis 
from scratch.
On the 
other hand, it should be noted that the RS fading model is a special case of the  $\kappa-\mu$ Shadowed fading model
when $\mu=1$ and $\kappa = K$ \cite{Paris2014}, which has been extensively 
investigated in the last few years, so that any performance metric already obtained for the $\kappa-\mu$ Shadowed can also be used
to obtain the corresponding metric for the FTR model. 

The use of finite-range integrals is common in communication theory, two compelling examples being the MGF approach in \cite{simon} for error probability analysis, and the proper-integral forms for the Gaussian Q-function \cite{hoyt18}, the Marcum Q-function \cite{hoyt19}, and the Pawula F-function \cite{hoyt20}. Additional finite-range integral formulations have been used to establish the connection between Hoyt and 
Rayleigh fading in \cite{NewF} and between TWDP and Rician fading models \cite{mgf}.  The integral connection between channel models with an arbitrary number $N$ of specular components plus a diffuse components 
(which includes the FTR model for $N=2$) and the RS fading model was first identified in \cite{Romero19}, but a standard procedure for leveraging known results for RS fading was not explored.

As an example of applying the derived statistical functions, we obtain outage probability expressions in a number of scenarios considering FTR fading, co-channel interference (CCI) and background noise, as well as when assuming interference-limited scenarios. Numerical results are presented to confirm the accuracy of our analytical derivations, showing the effects of different values of the parameters on the system performance.


The rest of this paper is organized as follows: In Section~II, the preliminary definitions are introduced.
Section~III and IV present the connections between the FTR and, respectively, RS and Nakagami-$m$ fading models, based on which the alternative formulations of the FTR model are carried out. 
Based on these connections, an outage probability analysis is performed in Section~V when the desired signal experiences FTR fading
and in the presence of interference.
Finally, Section~VI provides numerical and simulation results followed by the conclusions in Section~VII.

\section{Preliminary Definitions}
In the FTR fading model, the wireless channel consists of two fluctuating dominant waves, referred to as specular components, 
to which other diffusely propagating waves are added. The complex baseband received signal can be expressed as 
\begin{align}
	V = \sqrt{\zeta} V_1 \exp (j \phi_1) + \sqrt{\zeta} V_2 \exp (j \phi_2) +X+jY,\label{Vr}
\end{align}
where $ V_n $ and $\phi_n$, for $n=1,2$, represent, respectively, the average amplitude and
the uniformly distributed random phase of the n-\emph{th} specular component, such that $\phi _n \sim \mathcal{U}[0,2\pi)$. 
The term $X + jY$ is a complex Gaussian random variable, with $X,Y \sim \mathcal{N}(0,\sigma^2)$, representing the diffuse 
received signal component 
due to the combined reception of numerous weak scattered waves.
On the other hand, $\zeta$ is a unit-mean Gamma distributed random variable modulating the specular components, whose PDF is 
given by
\begin{align}
	f_{\zeta}(u) =\frac{m^m u^{m-1}}{\Gamma(m)} e ^{-mu}.
\end{align}
This model is conveniently expressed in terms of the parameters $K$ and $\Delta$, defined as
\begin{align}
	&\Delta =\frac{2V_1 V_2}{V_1^2+V_2^2},\\
	&K= \frac{V_1^2 +V_2^2}{2\sigma^2},
\end{align}
where $K$ denotes the power ratio of the specular components to the diffuse components, and 
$\Delta$ provides a measure of the similarity of the two specular components, ranging from 0 (one specular component is absent)
to 1 (the specular components have the same amplitude).
Additionally, we define the ancillary random variable $\theta \triangleq \phi_1-\phi_2$. Note that, as the phase difference is 
modulo $2\pi$, $\theta$ will be uniformly distributed, and therefore we can write $\theta \sim 
\mathcal{U}[0,2\pi)$.

Throughout this paper, we will characterize the distribution of the instantaneous SNR at the receiver, which will be denoted
by {$\gamma={(E_s/\overline N_0)\left|V\right|^2}$, with $E_s$ and $\overline N_0$ representing, respectively, the energy density per symbol and the power spectral density.

\textbf{Definition 1:}
A random variable $\gamma$ following a FTR distribution with parameters $m$, $K$, $\Delta$ and mean $\overline\gamma$ will be denoted by 
$\gamma \sim \mathcal{FTR}(\overline{\gamma},m,K,\Delta)$, and its PDF will be denoted by $f_{\gamma}^{\rm FTR}(x;\overline{\gamma},m,K,\Delta)$, where the 
parameters may be dropped from the notation when there is no confusion.

\textbf{Definition 2:}
A random variable $\gamma$ following a squared RS distribution with parameters $m$, $K_r$ and $\overline{\gamma}$ will be  denoted  by 
$\gamma \sim \mathcal{RS} (\overline{\gamma},m,K_r)$, and its PDF can be written as
\begin{align}\nonumber \label{fRS}
	&f_\gamma ^{\rm RS}(x;\overline \gamma, m,{K_r}) = {\left( {\frac{m}{{m + {K_r}}}} \right)^m}\frac{{1 + {K_r}}}{{\overline \gamma }} \\
	&\times  \exp \left( { - \frac{{1 + {K_r}}}{{\overline \gamma }}x} \right) {_1}{F_1}\left( {m;1;\frac{{{K_r}\left( {1 + {K_r}} \right)}}{{\overline \gamma \left( {m + {K_r}} \right)}}x} \right),
\end{align}
where $_1 F_1(\cdot)$ is the confluent hypergeometric function of the first kind \cite[eq. (9.210.1)]{Gradsh}, $\overline{\gamma}=(E_s/\overline N_0)2\sigma^2 (1+K_r)$, and $K_r=\frac{\Omega}{2\sigma^2}$,  where $\Omega$ and $2\sigma ^2 $ are the powers of the specular and diffuse components, respectively.

\textbf{Definition 3:}
A random variable $\gamma$ following a squared Nakagami-$\hat{m}$ distribution is expressed as $\gamma \sim \mathcal{K}(\hat{\gamma},\hat{m})$, and 
its PDF, with  integer $\hat{m}$, is given by 
\begin{align}\label{fNAK}
	f_\gamma ^{\mathcal{K}}(x;\hat\gamma,\hat m)=\left( \frac{\hat m}{\hat\gamma} \right)^{\hat m} \frac{x^{\hat m-1}}{(\hat m-1)!}e^{-x \hat m/{\hat\gamma}}.
\end{align}

\section{FTR Formulation as a Continuous Mixture of RS Variates}

In this section, we explore the connection of the FTR model with the RS fading model and present 
a new formulation of the PDF, the cumulative distribution function (CDF) and the moments for the FTR fading 
channel which are valid for arbitrary real $m$. We also 
present results for relevant Laplace-domain statistics that have not been previously investigated for FTR, such as the GMGF and IMGF.

\subsection{PDF and CDF of the FTR fading}
Expressions for the PDF and CDF of the FTR model as a continuous mixture of RS variates are now derived. As a direct implication of this connection, we will show that
any performance metric for FTR fading can be readily obtained from its counterpart for RS fading. 
\begin{lemma}\label{lemRSFTR}
Let  $\gamma$  be a random variable such that $\gamma \sim \mathcal{FTR}(\overline{\gamma};m,K,\Delta)$, then $\gamma$ is a 
continuous mixture of squared RS variates, whose PDF is given as
	\begin{align}
		f_\gamma ^{\rm FTR}(x;\overline{\gamma},&m,K,\Delta ) \nonumber \\& 
		= \frac{1}{\pi}\int_0^{\pi } {f_{\gamma | \theta}^{\rm RS}(x;\overline{\gamma},m,K\left( {1 + \Delta \cos (\theta )} \right)d\theta }, \label{RS2FTR}
		\end{align}
\end{lemma}
with $\gamma | \theta \sim \mathcal{RS} (\overline{\gamma},m,K\left( {1 + \Delta \cos (\theta )} \right))$.
\begin{proof}
For a particular realization of the random variable $\zeta$, the channel model defined in \eqref{Vr} corresponds to  
the TWDP model. Therefore, the PDF of the SNR for the FTR model can be written as
\begin{align}
	f_\gamma ^{\rm FTR}(x;m,K,\Delta ) = 
	{{E}_\zeta }\left[ {f_{\gamma | \zeta} ^{\rm TWDP}(x;\zeta K,\Delta )} \right], \label{twdp1}
\end{align}
where $E_X[\cdot]$ denotes the expectation operator over the random variable $X$ and $f^{TWDP}$ denotes the PDF of the TWDP model. 

Additionally, when there is only one specular component in \eqref{Vr} the FTR model collapses to the RS fading model, from which the Rice fading model is
obtained for a particular realization of $\zeta$. Thus, denoting the PDF of the Rice model as $f^{Rice}$, we can write
\begin{align}
	f_\gamma ^{\rm RS}(x;m,K_r ) = 
	{{E}_\zeta }\left[ {f_{\gamma | \zeta} ^{\rm Rice}(x;\zeta K_r )} \right]. \label{Rice-RS}
\end{align}

On the other hand, it was shown in \cite{mgf} that the TWDP and the Rice fading models are related by
	\begin{align}
	f_\gamma ^{\rm TWDP}(x;\hat {K},\Delta ) = {E_\theta }\left[ {f_{\gamma | \theta } ^{\rm Rice}(x;\hat {K}\left( {1 + \Delta \cos (\theta )} \right)} \right].\label{twdp2}
\end{align}
Consequently, using \eqref{twdp1} and \eqref{twdp2}, the PDF of the FTR model can be written as
\begin{align}\nonumber
	f_\gamma ^{\rm FTR}(x;m,K,&\Delta ) = {E_\zeta }\left[ {{E_\theta }\left[ {f_{\gamma | \theta, \zeta } ^{\rm Rice}(x;\zeta K\left( {1 + \Delta \cos (\theta )} \right)} \right]} \right]  \\
& = {E_\theta }\left[ {{E_\zeta }\left[ {f_{\gamma | \theta, \zeta } ^{\rm Rice}(x;\zeta K\left( {1 + \Delta \cos (\theta )} \right)} \right]} \right],
\end{align}
and taking into account \eqref{Rice-RS} we have
\begin{align}
		f_\gamma ^{\rm FTR}(x;&m,K,\Delta )
 = {E_\theta }\left[ {f_{\gamma | \theta} ^{\rm RS}(x;m,K\left( {1 + \Delta \cos (\theta )} \right)} \right] \nonumber
  \\&
  = \frac{1}{2\pi}\int_0^{2\pi } {f_{\gamma | \theta} ^{\rm RS}(x;m,K\left( {1 + \Delta \cos (\theta )} \right)d\theta }. 
\end{align}
Considering the symmetry of the cosine function around $\pi$, the integration with respect to $\theta$ can be peformed in $[0,\pi)$, 
thus obtaining \eqref{RS2FTR}.
 \end{proof}
\begin{remark} \label{rm1}
It must be noted that the factor $\frac{{\overline \gamma }}{{1 + K}}$ remains invariant in the transformation defined in 
\eqref{RS2FTR}, and this is also
true for all the transformations defined in the proof of Lemma \ref{lemRSFTR} in all subsequent expressions. This is 
due to the fact that
parameters $K$ and $	\overline \gamma $ are related by the expression
$	\overline \gamma  = \left( {{E_s}/{\overline N_0}} \right)\left( {V_1^2 + V_2^2 + 2{\sigma ^2}} \right) = \left( {{E_s}/{\overline N_0}} \right)2{
\sigma ^2}\left( {1 + K} \right)$, 
therefore, if $K$ varies  as a function of $\zeta$ and/or $\theta$, then $	\overline \gamma $ also varies as 
$
	\overline \gamma (\zeta ,\theta ) =\left( {{E_s}/{\overline N_0}} \right)2{\sigma ^2}\left( {1 + K (\zeta ,\theta )} \right),
$
yielding,
\begin{align}
\frac{{\overline \gamma }}{{1 + K}} = \left( {{E_s}/{\overline N_0}} \right)2{\sigma ^2} 
= {\rm{ }}\frac{{\overline \gamma (\zeta ,\theta )}}{{1 + K(\zeta ,\theta )}}.\label{eqg}
	\end{align}
\end{remark}

\begin{corollary}
 Let $\gamma \sim \mathcal{FTR}(\overline{\gamma},m,K,\Delta)$, then, its PDF can be computed as
\begin{align}
	\nonumber
	&	 f_\gamma ^{\rm FTR}(x;\overline{\gamma},m,K,\Delta )
	\\ \nonumber& =\frac{1}{\pi} \int_0^{\pi } {{\left( {\frac{m}{{m + K\left( {1 + \Delta \cos (\theta )} \right)}}} \right)}^m} \exp \left( { - \frac{{1 + K}}{{\overline \gamma }}x} \right)  \nonumber\\ &
	\times \frac{{1 + K}}{{\overline \gamma }}	{_1}{F_1}\left( {m;1;\frac{{\left( {1 + K} \right)K\left( {1 + \Delta \cos (\theta )} \right)}}{{\overline \gamma m + \overline \gamma K\left( {1 + \Delta \cos (\theta )} \right)}}x} \right)d\theta. \label{fFTR}
\end{align}
\end{corollary}
\begin{proof}
	This expression is obtained by plugging \eqref{fRS} into (\ref{RS2FTR}) and considering \eqref{eqg}.
\end{proof}

It was demonstrated in \cite{ftr} that the FTR fading model collapses to the Hoyt fading model for $m=1$. In this regard,
it can be easily shown that the integral connection
between the Hoyt and Rayleigh models presented in \cite{NewF} is actually a particular case of \eqref{fFTR} when $m=1$.
\begin{remark} \label{rm2}
Lemma \ref{lemRSFTR} implies that, conditioning on $\theta$, the FTR distribution is actually a RS distribution, and we can write
the conditional PDF as
\begin{align}\nonumber
	&	f_{\gamma | \theta }^{\rm FTR}(x;m,K,\Delta ) = f_{\gamma | \theta}^{\rm RS}(x;m,K(1+\Delta cos (\theta)) ) =
	\\ \nonumber &
	{\left( {\frac{m}{{m + K\left( {1 + \Delta \cos (\theta )} \right)}}} \right)^m}\frac{{1 + K}}{{\overline \gamma }}
	\exp \left( { - \frac{{1 + K}}{{\overline \gamma }}x} \right)
	\\& \times  \;{_1}{F_1}\left( {m;1;\frac{{K\left( {1 + K} \right)\left( {1 + \Delta \cos (\theta )} \right)}}{{\overline \gamma m +
 K\overline \gamma \left( {1 + \Delta \cos (\theta )} \right)}}x} \right)\label{fFTRc},
\end{align}
which will be used in subsequent derivations.
\end{remark}

\begin{lemma} \label{lemcdf}
	 Let $\gamma \sim \mathcal{FTR}(\overline{\gamma},m,K,\Delta)$, then, its CDF can be calculated as
\begin{align}\nonumber
	&	F_{\gamma}^{\rm FTR}(x;m,K,\Delta ) \\& \nonumber
	=\frac{1}{\pi} \int_0^{\pi} \frac{{1 + K}}{{\overline \gamma }}x{\left( {\frac{m}{{m + K\left( {1 + \Delta \cos (\theta )} \right)}}} \right)^m} \times \\
	&   {\Phi _2}\left( {1 - m, m;2; - \frac{{1 + K}}{{\overline \gamma }}x; \frac{{-\left( {1 + K} \right)mx}}{{\overline \gamma (m + K \left( {1 + \Delta \cos (\theta )} \right))}}} \right) d\theta,
	\label{CDF}
\end{align}
where $\Phi_2$ is the bivariate confluent hypergeometric function defined in \cite[p. 34, eq. (8)]{phi2}.
\end{lemma}
\begin{proof}
The CDF of the FTR can be written from the CDF of the RS that was obtained in \cite[eq. (8)]{RSParis} as follows:
\begin{align}\nonumber
&	F_\gamma ^{\rm RS}(x ) = \frac{{1 + K}}{{\overline \gamma }}x{\left( {\frac{m}{{m + K}}} \right)^m}\\
	&  \times {\Phi _2}\left( {1 - m,m;2; - \frac{{1 + K}}{{\overline \gamma }}x; - \frac{{\left( {1 + K} \right)mx}}{{\overline \gamma \left( {m + K} \right)}}} \right),
\end{align}
then, it is clear from Remarks 1 and 2 that the conditional CDF of the FTR model can be written as		
\begin{align}
&	F_{\gamma \left| \theta  \right.}^{\rm FTR}(x;m,K,\Delta )= \nonumber
 F_\gamma ^{\rm RS}(x;K\left( {1 + \Delta \cos (\theta )} \right))\\ \nonumber
 & = \frac{{1 + K}}{{\overline \gamma }}x{\left( {\frac{m}{{m + K\left( {1 + \Delta \cos (\theta )} \right)}}} \right)^m}\times\\
	&   {\Phi _2}\left( {1 - m, m;2; - \frac{{1 + K}}{{\overline \gamma }}x; - \frac{{\left( {1 + K} \right)mx}}{{\overline \gamma m + K\overline \gamma \left( {1 + \Delta \cos (\theta )} \right)}}} \right).
	\end{align}
By integrating over the parameter $\theta$, the proof is completed.
\end{proof}

Based on the connection between the FTR and the RS distributions, we show in the following lemma that the 
performance metrics of FTR can be derived from the metrics of the RS fading model. 
\begin{lemma} \label{lemRSFTRmet}
Let $h(\gamma)$ be a performance metric (or function) depending on the instantaneous SNR $\gamma$, and let 
$X^{\rm RS} (\overline \gamma,m,K)$ be the metric (or statistical function) in RS fading obtained by averaging over an interval of the PDF 
of the SNR, i.e.,
\begin{equation} \label{eq:008}
	 X^{\rm RS} (\overline \gamma,m,K)= \int_a^b 
	 h(x) f_\gamma ^{\rm RS}(x;\overline \gamma, m,{K}) 	 dx,
\end{equation}
where $0 \le a\le b < \infty$. Then, the average performance metric in FTR fading channels can be calculated as
	\begin{align}\label{xftr} \nonumber
		X^{\rm FTR}(\overline \gamma,&m,K,\Delta) \\& =\frac{1}{\pi}\int_0^{\pi} X^{\rm RS}(\overline \gamma,m,K(1+\Delta \cos(\theta))) d\theta,
	\end{align}
	where Remark \ref{rm1} must be taken into account.
		\end{lemma}
\begin{proof}
The average metric in a FTR fading channel will be calculated as
\begin{align}
  X^{\rm FTR}(\overline \gamma,m,K,\Delta)=\int_a^b h(x)f^{\rm FTR}_\gamma (x;\overline \gamma, m,K,\Delta) dx. \label{xrs}
\end{align}
The result is obtained by plugging \eqref{RS2FTR} into \eqref{xrs} and changing the order of  integration. 
\end{proof}

Lemma \ref{lemRSFTRmet} has important consequences, as the RS model has a simpler formulation than the FTR model and different
performance metrics for the latter are available in the literature in closed-form. 
 
 
\subsection{MGF of the FTR fading  }
As a consequence of Remark \ref{rm2}, the conditional MGF of the FTR distribution can be written as
	\begin{align}
		&	M_{\gamma|\theta}^{\rm FTR}(s;m,K,\Delta)= M_{\gamma}^{\rm RS}(s;m,K(1+\Delta \cos(\theta )),
	\end{align}
	where $M_{\gamma}^{\rm RS}$ is the MGF of the squared RS distribution, which is given by \cite{Abdi2003}:
	\begin{align}
		M_\gamma ^{\rm RS}(s;m,K,\Delta ) = \frac{{{m^m}{{\left( {1 - \frac{{\overline \gamma }}{{1 + K}}s} \right)}^{m - 1}}}}{{{{\left( {m - \left( {m\frac{{\overline \gamma }}{{1 + K}} + \frac{{\overline \gamma }}{{1 + K}}K} \right)s} \right)}^m}}}\label{mgfRS}.
	\end{align}
	Therefore, noting that $\frac{\overline\gamma}{1+K}$ remains invariant, as justified in Remark \ref{rm1},  the unconditional MGF of the FTR model will be
	\begin{align}\nonumber
		&	M_{\gamma}^{\rm FTR}(s;m,K,\Delta) = \frac{1}{\pi}\int_0^{\pi} M_{\gamma|\theta}^{\rm FTR}(s;K,\Delta) d\theta
		\\&
		=\frac{1}{\pi}\int_0^{\pi} \frac{(1+K)m^m (1+K-\overline{\gamma}s)^{m-1}}{((1+K)m-(m+K(1+\Delta \cos(\theta)))\overline{\gamma}s)^m}d\theta.\label{MFTR}
	\end{align}
	This integral can be solved in closed-form, as shown in the Appendix as integral I1, yielding
		\begin{align}\nonumber
		M_\gamma ^{\rm FTR}(s;m,K,\Delta )& = \frac{{\left( {1 + K} \right){m^m}{{\left( {1 + K - \overline \gamma s} \right)}^{m - 1}}}}{{\sqrt R }}
		\\ &\times
		{P_{m - 1}}\left( {\frac{{\left( {1 + K} \right)m - \left( {m + K} \right)\overline \gamma s}}{{\sqrt R }}} \right),\label{mgfftr}
	\end{align}
	where $P_{m-1}(.)$ is the Legendre function of the first kind of degree $m-1$ and $R$ is defined as
	\begin{align}
		&	R \triangleq  {\left( {\left( {1 + K} \right)m - \left( {m + K} \right)\overline \gamma s} \right)^2} - {\left( {K\overline \gamma \Delta s} \right)^2} \nonumber \\
		&	= \left[ {{{\left( {m + K} \right)}^2} - {{\left( {K\Delta } \right)}^2}} \right]{{\overline \gamma }^2}{s^2} - 2\left( {1 + K} \right)m\left( {m + K} \right)\overline \gamma s \nonumber \\
		& + {\left[ {\left( {1 + K} \right)m} \right]^2}.\label{Rfunc}
	\end{align}
The MGF expression given in \eqref{mgfftr}-\eqref{Rfunc} coincides  with the one previously obtained in \cite[eqs. (8)-(9)]{ftr}, 
which, by virtue of the uniqueness theorem of the MGF, confirms that the PDF given in \eqref{fFTR} is a generalization, for arbitrary $m$, of the PDF given in 
\cite[eq. (15)]{ftr}.

\subsection{Generalized MGF and moments}
The generalized MGF is an important statistical function relevant to wireless communication theory, as it 
naturally appears when analyzing different scenarios such as energy detection probability, outage probability with co-channel 
interference in interference limited scenarios, physical layer 
security analysis \cite{beckmann}, or in the context of composite fading channel modeling \cite{Pablo}.
The GMGF of a random variable $X$ is defined: 
\begin{align}
&{M^{(n)}_X}(s) = \int_0^\infty  {{x^n}\exp \left( {xs} \right)f_X(x)dx}\label{gmgfd}.
\end{align}
Note that in the case of $n\in \mathbb{N}$, the generalized MGF coincides with the $n^{th}$ order derivative 
of the MGF.
\begin{lemma}\label{lemGMGF}
	Let  $\gamma \sim \mathcal{FTR}(\overline{\gamma};m,K,\Delta)$, then, the GMGF of $\gamma$ can be obtained in closed-form
as (\ref{gmgf}),
 \begin{figure*} 
\begin{align}\nonumber
	M_{\gamma}^{(n)}(s) &= n!{{{m^m}{{\left( {1 + K - \overline \gamma s} \right)}^{m-n-1}}}}
	{{\overline \gamma }^n}\sum\limits_{l = 0}^n {\binom{n}{l}\frac{{{{\left( m \right)}_l}}}{{l!}}} 
	\frac{{( {1 + K} )}^{l + 1}K^l}{{{{\left[ {m\left( {1 + K} \right) - \left( {m + K - K\Delta } \right)\overline \gamma s} 
\right]}^{m+l}}}}
	\sum\limits_{q = 0}^l 
	\binom{l}{q}
	{
		\left( {1 - \Delta } \right)^{l - q}}
	{\left( {2\Delta } \right)^q}
		\nonumber\\ 	& 
	\times 	\frac{{\Gamma \left( {\frac{1}{2} + q} \right)\Gamma \left( {\frac{1}{2}} \right)}}{{\pi\Gamma (q + 1)}} 
	 {{\kern 1pt}_2}{F_1}\left( {{m+l},\frac{1}{2} + q;q + 1;\frac{{2K\Delta \overline \gamma s}}{{m\left( {1 + K} \right) - \left
( {m + K - K\Delta } \right)\overline \gamma s}}} \right).\label{gmgf}
	 \\
	\hline \nonumber
\end{align}
 \end{figure*}
where $_2F_1(\cdot)$ is the Gaussian hypergeometric function \cite[eq. (9.100)]{Gradsh} and where $(a)_n$ is the Pochhammer symbol.
\end{lemma}
\begin{proof}
The conditional generalized MGF 
 can be written by plugging (\ref{fFTRc}) into \eqref{gmgfd}, i.e.:
\begin{align}\nonumber
	{M^{(n)}_{\gamma|\theta}}(s) &= 
\int_0^\infty  {x^n}\exp \left( {xs} \right)
			{\left( {\frac{m}{{m + K\left( {1 + \Delta \cos (\theta )} \right)}}} \right)^m}
				\\& \nonumber
		\times	\frac{{1 + K}}{{\overline \gamma }}
	\exp \left( { - \frac{{1 + K}}{{\overline \gamma }}x} \right)
		\\&  \times \;{_1}{F_1}\left( {m;1;\frac{{K\left( {1 + K} \right)\left( {1 + \Delta \cos (\theta )} \right)}}{{\overline \gamma m + K\overline \gamma \left( {1 + \Delta \cos (\theta )} \right)}}x} \right)
				dx. \label{mgfc}
\end{align}
The integral in \eqref{mgfc} can be solved with the help of \cite[eq. (7.621.4)]{Gradsh} as
\begin{align}\nonumber
&	M_{\gamma \left| \theta  \right.}^{(n)}(s)= \\ \nonumber
&  {\left( {\frac{m}{{m + K + K\Delta \cos (\theta )}}} \right)^m}\frac{{1 + K}}{{{{\left( {1 + K - \overline \gamma s} \right)}^{n +
 1}}}} 	{{\overline \gamma }^n} \Gamma (n + 1)
	\\& \nonumber
{_2}{F_1}\left( {m,n + 1;1;\frac{{\left( {1 + K} \right)\left( {K + K\Delta \cos (\theta )} \right)}}{{\left( {m + K + K\Delta 
\cos (\theta )} \right)\left( {1 + K - \overline \gamma s} \right)}}} \right),\label{mmgfn}
\\ \nonumber
& = {\left( {\frac{m(1+K-\overline{\gamma}s)}{m(1+K)-({m + K + K\Delta \cos (\theta )})\overline{\gamma} s}} \right)^m}
\\&\times 
\frac{({1 + K})\Gamma (n + 1)}{{{{\left( {1 + K - \overline \gamma s} \right)}^{n + 1}}}}  
	{{\overline \gamma }^n} \times
\nonumber	\\ 
	& 
	 \;{_2}{F_1}\left( {m,-n;1;
		-\frac{{\left( {1 + K} \right)\left( {K + K\Delta \cos (\theta )} \right)}}{{\left( {m + K + K\Delta \cos (\theta )} \right)
\left( {1 + K - \overline \gamma s} \right)}}} \right),
\end{align}
where the last equality is obtained with the help of  \cite [eq. (9.131)]{Gradsh}. 
Besides, based on 
\cite[p.17 eq. (12)]{phi2}
for integer $n$, the hypergeometric function can be written as
\begin{align} \nonumber
		_2{F_1}\left( { m,-n;c;z} \right)&=
	\;	_2{F_1}\left( { - n,m;c;z} \right) \\&= 
		\sum\limits_{l = 0}^n {{{( - 1)}^l}} \binom{n}{l}
		\frac{{{{\left( m \right)}_l}}}{{{{\left( c \right)}_l}}}{z^l},
\end{align}
thus, the conditional GMGF is determined as
\begin{align}\nonumber
M_{\gamma \left| \theta  \right.}^{(n)}(s)& = n!\frac{{{m^m}{{\left( {1 + K - \overline \gamma s} \right)}^m}}}{{{{\left( {1 + K - \overline \gamma s} \right)}^{n + 1}}}}{{\overline \gamma }^n}\sum\limits_{l = 0}^n {\binom{n}{l}\frac{{{{\left( m \right)}_l}}}{{l!}}} 
\\&
\times\frac{{{{\left( {1 + K} \right)}^{l + 1}}K^l{{\left( {1 + \Delta \cos (\theta )} \right)}^l}}}{{{{\left[ {m\left( {1 + K} \right) - \left( {m + K + K\Delta \cos (\theta )} \right)\overline \gamma s} \right]}^{l + m}}}}.\label{gmgfC}
\end{align}
Therefore, the unconditional GMGF will be given by
\begin{align}\nonumber
	M_{\gamma}^{(n)}	= \frac{1}{\pi} \int_0^{\pi}  M_{\gamma \left| \theta  \right.}^{(n)}(s) d \theta,
\end{align}
and with the help of integral I2 solved in the Appendix the proof is completed. 
\end{proof}

It should be noted that \eqref{gmgf} is valid for all possible values of the channel parameters when 
$s$ is a non-positive real number (which is usually the case in communication theory applications),
 since for the Gaussian hypergeometric function $_2 F_1(a,b;c;z)$ in the expression it is always true that $z \in \mathbb{R}$ with $z \leq 0$,
which permits its computation using its integral representation \cite[eq. (9.111)]{Gradsh}.

The moments of the SNR can be readily obtained from the GMGF, as we now show.
\begin{corollary}\label{moment}
Let  $\gamma \sim \mathcal{FTR}(\overline{\gamma};m,K,\Delta)$, then, the moments of $\gamma$ will be given by
\begin{align}\nonumber 
   	{\mu ^n} & 
  =
  {n!}
  {\left( {\frac{{\overline \gamma }}{{1 + K}}} \right)^n}\sum\limits_{l = 0}^n { \binom{n}{l} \frac{{{{K^l\left( m \right)}_l}}}{{l!m^l
  		}
  }}
  \\ & \times
   \, \sum_{q=0}^l (-1)^q \binom{l}{q} \frac{\Gamma(\frac{1}{2}+q)}{\sqrt{\pi}\Gamma(1+q)}\left( {2 \Delta}\right)^q (1-\Delta)^{l-q}. \label{moments}
	\end{align}
\end{corollary}

\begin{proof}
	This result is obtained by simply substituting $s=0$ in (\ref{gmgf}). 
\end{proof}

The moments of the FTR model where presented in \cite{moment}, however, different expressions were given for different values of the parameter $\Delta$.
On the other hand, \eqref{moments} is valid for any $\Delta$.

\subsection{Incomplete MGF of the FTR model}
The incomplete MGF is widely used in different situations for analyzing the performance of communication 
systems, therefore it has prominent relevance in communication theory.
These situations include order statistics, symbol and bit error rate calculations, 
capacity analysis in fading channels, outage probability in cellular systems, adaptive scheduling techniques, cognitive relay 
networks, or physical layer security \cite{IMGF}.

The lower and upper IMGF of a random variable $X$ are defined by changing one of the limits of integration in the definition of the MGF by a non-negative real number 
$\delta\in(0, \infty)$, as follows \cite{IMGF}:
\begin{align}
	\mathcal{M}_X^l(s,\delta)\overset{\Delta}{=}\int_0^{\delta} e^{sx}f_X(x)dx
\end{align}
\begin{align}
	\mathcal{M}_X^u(s,\delta)\overset{\Delta}{=}\int_{\delta}^{\infty} e^{sx}f_X(x)dx.
\end{align}
The MGF of $X$ is related to the lower and upper MGF by the relations: 
\begin{align}
M_X(s)	= \mathcal{M}_X^l(s,\infty)=\mathcal{M}_X^u(s,0),
\end{align}
\begin{align}
\mathcal{M}_X^u(s,\delta)=	M_X(s)	- \mathcal{M}_X^l(s,\delta).\label{Mu}
\end{align}
In the following lemma, an expression of the lower incomplete MGF is provided for the FTR model.
\begin{lemma}
		Let $\gamma \sim \mathcal{FTR}(\overline{\gamma},m,K,\Delta)$, then, its lower 
IMGF is given in (\ref{Limgf}) (see the top of the next page), where $\Phi_2$ is the bivariate confluent hypergeometric 
function.  
\begin{figure*}
\begin{align}\label{Limgf}
	M_{\gamma}^{l}(s,z) = \frac{1}{\pi}\int_0^{\pi} \frac{m^m(1+K)z}{\overline{\gamma}(m+K(1+\Delta \cos(\theta)))^m}  
 \Phi_2 \left(1-m,m,2,\left(s-\frac{1+K}{\overline{\gamma}}\right)z, \left(s-\frac{m(1+K)}{\overline{\gamma}(m+K+K\Delta \cos(\theta))}
\right)z \right) d\theta.
 \\
 \hline \nonumber
 \end{align}
\end{figure*}
\end{lemma}
\begin{proof}
This result is obtained by considering Lemma \ref{lemRSFTRmet} and using the lower IMGF of the RS model presented 
in \cite[eq. (13)]{IMGF}.
\end{proof}

  \section{FTR Formulation as a Continuous Mixture of Nakagami-$m$ Variates}
 In this section, we describe a formulation to connect the Nakagami-$m$ model to the FTR fading for integer $m$. In this 
regard, we introduce the PDF, CDF and the incomplete generalized MGF of the FTR fading model in terms of finite-range integrals of elementary functions.
Note that the GMGF, the IMGF, and the moments are particular cases of the IGMGF, from which they can easily be obtained.
We also show that any given metric already known for Nakagami-$m$ can be readily extended to the FTR case when $m$ is an integer.

  \begin{lemma}\label{lemnakaftr} 
 Let  $\gamma$  be a random variable such that $\gamma \sim \mathcal{FTR}(\overline{\gamma};m,K,\Delta)$ with $m \in \mathbb{N}$, then $\gamma$ is a 
continuous mixture of squared Nakagami-$m$ variates, which PDF and CDF are given, respectively, as
 	\begin{align}\nonumber
 		f_\gamma&^{\rm FTR}(x;\overline\gamma,m,K,\Delta)=
 		\\ &\frac{1}{\pi}\int_0^{\pi} \sum_{i=0}^{m-1}  C_i(\theta) f_{\gamma | {\theta,i}}^{\mathcal{K}}(x; (m-i)\Omega(\theta),m-i) d\theta,\label{fNak}
		\\ 
		F_\gamma&^{\rm FTR}(x;\overline\gamma,m,K,\Delta)= \nonumber
 		\\ & 1-\frac{1}{\pi}\int_0^{\pi} \sum\limits_{i = 0}^{m - 1 } {C_i(\theta) } e^{ - x/\Omega(\theta) } \sum\limits_{r = 0}^{m - i - 1} {\frac{1}
{{r!}}\left( {\frac{x}
{{\Omega(\theta)  }}} \right)^r } d\theta, \label{cdfNak}
 	\end{align}
 		where $\gamma | {\theta,i} \sim \mathcal{K}(x; (m-i)\Omega(\theta),m-i)$ and
 	\begin{align}\label{Ci}
 		&C_i(\theta)=\binom{m-1}{i}  \frac{m^i\left( {K(1+\Delta \cos(\theta))}\right)^{m-i-1}}{\left(m+K(1+\Delta \cos(\theta))\right)^{m-1}},
 		\\
 		&\Omega(\theta)=\frac{\overline\gamma}{1+K}\frac{m+K(1+\Delta \cos(\theta))}{m}.\label{Omega}
 	\end{align}
 \end{lemma}
 \begin{proof}
 	It was shown in \cite{kappamuinteger} that, when parameter $m$ is an integer, the PDF and CDF of the squared RS distribution can be expressed
	in terms of the statistics of the squared Nakagami-$m$.
	By recognizing that the RS fading model is a special case of the $\kappa-\mu$ Shadowed fading model when $\mu=1$ and 
$\kappa = K$ we can write:
 	\begin{align}\label{rsnak}
 		&f^{\rm RS}(x;\overline\gamma,m,K)=\sum_{j=0}^{m-1} B_j f^{\mathcal{K}}(x;(m-j)\Omega_B;m-j), \\&
		F^{\rm RS}(x;\overline\gamma,m,K) = 1 - \sum\limits_{1 = 0}^{m - 1 } {B_j^{} } e^{ - x/\Omega _B } \sum\limits_{r = 0}^{m - j - 1} {\frac{1}
{{r!}}\left( {\frac{x}
{{\Omega _B }}} \right)^r } ,
 	\end{align}
 	where 	 	
	\begin{align}\label{bj}
		&B_j  = \left( {\begin{array}{c}
   {m - 1 }  \\ 
   j  \\ 
 \end{array} } \right) {\frac{m^j  K^{m - 1  - j}}
{{(m  + K)}^{m - 1 }}}   ,  \\
&\Omega_B=\frac{\overline\gamma}{1+K}\frac{m+K}{m} .\label{omega}
		\end{align}
		
 	If we now apply Lemma~\ref{lemRSFTRmet}, the PDF and CDF of the FTR model are obtained.
 \end{proof}
 
We now demonstrate that any performance metric already known for Nakagami-$m$ fading with integer $m$ can readily be extended to FTR fading
by applying a finite range integral to the Nakagami-$m$ metric.

\begin{lemma}\label{lemmakaftr}
Let $h(\gamma)$ be a performance metric (or function) depending on the instantaneous SNR $\gamma$, and let 
${{X}^{\mathcal{K}}}(\overline\gamma,m)$ be the metric (or statistical function) under Nakagami-$m$  fading with average 
SNR $\overline\gamma$ and integer $m$ obtained by averaging over an interval of the PDF of the SNR, i.e.,
\begin{equation} \label{eq:0082}
	 X^{{\mathcal{K}}} (\overline \gamma,m)= \int_a^b 
	 h(x) f_\gamma ^{{\mathcal{K}}}(x;\overline \gamma, m,) 	 dx,
\end{equation}
where $0 \le a\le b < \infty$. Then, the average performance metric in FTR fading channels with average SNR 
$\overline\gamma$ and integer $m$ can be calculated as
 	\begin{align}\nonumber
 		&	X^{\rm FTR}(\overline\gamma,m,K,\Delta)
 		\\&=\frac{1}{\pi}\int_0^{\pi} \sum_{i=0}^{m-1} C_i(\theta) X^{\mathcal{K}}\left((m-i){\Omega(\theta)},m-i\right)d\theta.
 	\end{align}
 
\end{lemma}
 \begin{proof}
 	From Lemma~\ref{lemRSFTRmet}, any metric or statistical function for the FTR model, $X^{\rm FTR}(\overline\gamma,m,K,\Delta)$, can be 
calculated in terms of the RS metric. By considering that, for  integer $m$, the statistics 
of the RS model can be written in terms of the Nakagami-$m$ model,  as shown in Lemma \ref{lemnakaftr}, the proof is completed. 

 \end{proof}

It is important to note that the Nakagami-$m$ model has been for decades a widely used stochastic fading model for wireless channels, for which
closed-form results are available for a myriad of performance metrics and system models. From Lemma \ref{lemmakaftr}, all those results can be readily extended in 
a straightforward manner to the much more general FTR fading model.

 
 \subsection{Incomplete generalized MGF of FTR}
The incomplete generalized MGF finds application in the outage probability calculation in the presence of co-channel interference
and background noise \cite{CCI} or in the secrecy outage when the legitimate link undergoes arbitrary fading \cite{IMGF}.
The (upper) incomplete generalized MGF of a random variable $X$ is defined as follows:
 \begin{align}
 	G_X(n,s,\Lambda) = \int_{\Lambda}^{\infty} x^n e^{sx} f_X(x) dx,
 \end{align}
where $f_X(x)$ is the PDF of $X$ and we assume $n\in \mathbb{N}$. 

\begin{lemma}\label{lemIGMGF}
	Let a random variable $\gamma \sim \mathcal{FTR}(\overline \gamma; m, K,\Delta)$ with $m \in \mathbb{N}$. Then, the IGMGF of $\gamma$  
is obtained as 
\begin{align}\nonumber
		&	G^{\rm FTR}_{\gamma}(n,s,\Lambda;\overline\gamma,m,K,\Delta)=
	\\&\frac{1}{\pi}\int_0^{\pi} \sum_{i=0}^{m-1} C_i(\theta) G^{\mathcal{K}}_{\gamma}\left(n,s,\Lambda;( m-i){\Omega(\theta)}, m-
i\right)d\theta,
	\end{align}
where $C_i(\theta) $ and $\Omega(\theta)$ are defined in (\ref{Ci}) and (\ref{Omega}), respectively, and	$G^{\mathcal{K}}_{\gamma}
\left(n,s,\Lambda;\hat \gamma,\hat m\right)$ is the IGMGF of the Nakagami-$m$ fading that was obtained in \cite[eq. (25)]{CCI}
which, using 	\cite[eq. 8.352.2]{Gradsh}, can be written as
\begin{align}\nonumber
	&G^{\mathcal{K}}_{\gamma}\left(n,s,\Lambda;\hat \gamma,\hat m\right) =\frac{(\hat m/\hat \gamma)^{\hat m}}{(\hat m/\hat \gamma
 -s)^{\hat m+n}}
\\ &
\times \frac{{\left( {\hat m + n - 1} \right)!}}
{{\left( {\hat m - 1} \right)!}}e^{ - \left( {\hat m /\hat\gamma  - s} \right)\Lambda } \sum\limits_{j = 0}^{\hat m + n - 1} {\frac{{\left( {\hat m/\hat\gamma  - s} \right)^j \Lambda ^j }}
{{j!}}} \label{NakIGMGF}. 
\end{align}
\end{lemma}
\begin{proof}
This result is obtained by the direct application of Lemma \ref{lemmakaftr} to \eqref{NakIGMGF}. 
\end{proof}

\section{Application Example:
Outage Probability of Multi-antenna receiver with CCI }

The statistical analysis carried out in the previous sections allows to analyze numerous wireless scenarios undergoing FTR 
fading, including those arising from the incomplete and generalized MGF formulations, as well as those available in the literature when assuming RS or Nakagami-$m$ fading.
As an example of application, in this section we investigate the outage probability in a  wireless
communication system in the presence co-channel interference (CCI) when the desired user experiences FTR fading and considering 
two different scenarios: $(A)$ the desired signal experiences FTR fading (with integer $m$) and the receiver suffers CCI and 
background noise; and $(B)$ the desired signal experience FTR fading (with arbitrary $m$) and the receiver is CCI limited, i.e., the background noise is assumed to be negligible. 

Consider a wireless communication system where the receiver is affected by $L$ interfering signals and additive white Gaussian 
noise~(AWGN). The received signal from the desired user and the interfering users undergo FTR and i.i.d. 
Rayleigh fading, respectively. Thus, the SINR at the receiver can be written as
\begin{align} \label{sinr}
	SINR =\frac{W}{Y+N_0},
	\end{align}
where $W$ is the received power from the desired signal, $Y$ is the sum of $L$ independent exponential random variables, and 
$N_0$ is the background noise power.
The outage probability of the SINR for the considered system  can be defined as 
\begin{align}
	P_{\rm out}= P\left(\frac{W}{Y+N_0}<R_{th}\right),\label{pcci}
		\end{align}
where $R_{th}$ denotes the SINR threshold. On the other hand, the CDF of the total interference power can be expressed as \cite[eq. (8)]{CCI}
	\begin{align}
	F_Y(y)=1- e^{-y/P_I}\sum_{k=0}^{L-1}\frac{1}{k!}\left(\frac{y}{P_I}\right)^k,\label{fy}
\end{align}
where $P_I$ denotes the received power for every CCI signal which, for simplicity, is assumed to be the same for all interferers. 
\subsection{CCI with background noise case with integer $m$}
The outage probability when the background noise cannot be neglected is given by
\begin{align}\nonumber
		P_{\rm out} =&F_W(R_{th}N_0)
	\\& +\int_{R_{th}N_0}^{\infty} \left[1-F_Y\left(\frac{w}{R_{th}}-N_0\right)\right] f_W(w) dw,\label{pout}
\end{align}
where $f_W(w)$ is the PDF of $W$. Note that all the statistical functions of $W$ are intimately related to functions of the SNR $\gamma$
derived in the previous sections, as $W = \frac{\gamma}{E_s/\overline N_0}$; and it is straightforward to show that,
in practice, they are obtained by considering $\overline W$ instead of $\overline \gamma$ in the expressions, with 
$\overline W = \frac{\overline \gamma}{E_s/\overline N_0}$.

From \eqref{fy} and \eqref{pout}, the outage probability can be written as \cite[eq. (21)]{CCI}
\begin{align}\nonumber\label{Pout}
	P_{\rm out} &
	= F_W(R_{th}N_0) +
	\sum_{k=0}^{L-1} \sum_{l=0}^k \frac{e^{N_0/P_I} (-N_0)^{k-l}}{l!(k-l)!{P_I}^k {R_{th}}^l}
	\\&\times
	G_W\left(l,-\frac{1}{R_{th}P_I},R_{th}N_0\right),
\end{align}
where $F_W(w)$ and $G_W(n,s,\Lambda)$ are, respectively, the CDF and IGMGF of the received power signal under FTR  fading, which can be directly
obtained for integer $m$ from \eqref{cdfNak} and \eqref{NakIGMGF}, respectively, yielding \eqref{PoutageNew}, which can be efficiently
 computed by simple finite-range integrations of elementary functions.

\begin{figure*}
	\begin{align}\nonumber\label{PoutageNew}
		P&_{\rm out} 
 = 1 - \frac{1}
{\pi }\sum\limits_{i = 0}^{m - 1} {\int_0^\pi  {C_i \left( \theta  \right)e^{ - R_{th} N_0 /\Omega \left( \theta  \right)} \sum\limits_{r = 0}^{m - i - 1} {\frac{1}
{{r!}}\left( {\frac{{R_{th} N_0 }}
{{\Omega \left( \theta  \right)}}} \right)^r } d\theta } } 
		+	\sum_{k=0}^{L-1} \sum_{l=0}^k \frac{e^{N_0/P_I} (-N_0)^{k-l}}{l!(k-l)!{P_I}^k {R_{th}}^l} 
					\\& 
		\times
		\frac {1}{\pi}\int_0^{\pi} \sum_{i=0}^{m-1}  C_i(\theta)
		{\left(\frac{1}{\Omega(\theta)}\right)^{m-i}}
\frac{{(m - i + l - 1)!}}
{{\left( {m - i - 1} \right)!}}e^{ - \left( {1/\Omega \left( \theta  \right) + 1/R_{th} P_I } \right)R_{th} N_0 } \sum\limits_{j = 0}^{m - i + l - 1} {\frac{{\left( {R_{th} N_0 } \right)^j }}
{{j!\left( {1/\Omega \left( \theta  \right) + 1/R_{th} P_I } \right)^{m - i + l - j} }}} 
		d\theta.
			\\
	\hline \nonumber
		\end{align}
\end{figure*}

	\subsection{Interference limited case with arbitrary $m$}
	We now consider a wireless communication system with $N$ receive antennas performing MRC with negligible background noise,
	i.e., the system is interference limited with $L$ i.i.d. Rayleigh interferers. In this case, $W$ in \eqref{sinr} represents the output signal at the MRC receiver
	from the desired user, which is assumed to undergo FTR fading with arbitrary $m$.
 The outage probability can be computed now as 
				\begin{align}\label{out-ILS}
			\hat P_{\rm out} = 1- \int_0^{\infty} F_Y \left(w/\hat R_{th}\right) f_W(w)dw,
		\end{align}
		where $\hat R_{th}$ is the SIR threshold. 
	By plugging (\ref{fy}) into \eqref{out-ILS}, we have 
		\begin{align}
		\hat P_{\rm out} = \sum_{k=0}^{L-1}\frac{1}{k!}\left(\frac{1}{\hat R_{th}P_I}\right)^k
		\int_0^{\infty}w^k e^{-w/\hat R_{th}P_I} f_W(w)dw.\label{poutint}
	\end{align}
The integral in (\ref{poutint}) represents the $k$-th order  generalized MGF of $W$, $M^{(k)}_W(s)$, when $s=\frac{-1}{\hat 
R_{th}P_I}$. Assuming the received power at every antenna to be affected by i.i.d. FTR fading with average $\overline W$, the MGF of $W$ at the MRC output
will be given by
\begin{align}
	M_W(s) = \prod_{i=1}^N M_{W_i}(s),
\end{align}
where $W_i$ is the received signal power at the $i$-th antenna. The $k$-th ordered GMGF of $W$ is determined as \cite[
eq. (13)]{CCI}
\begin{align}
	M^{(k)}_W(s) = \frac{d^k M_W(s)}{ds^k} = k! \sum_{\mathcal{U}} \prod_{i=1}^N \frac{M^{(u_i)}_{W_i}(s)}{u_i !},
\end{align}
where $\mathcal{U}$ is a set of $N$-tuples such that $\mathcal{U}=\{(u_1 ... u_N), u_i\in \mathbb{N},  ~\sum_{i=1}^N u_i = k\}$, 
 and $M^{(u_i)}_{W_i}(s)$ is the $u_i$-th order GMGF of the received signal power at the $i$-th antenna.
 Then, the  outage probability will be given by  \cite[eq. (15)]{CCI}
\begin{align}
\hat P_{\rm out} =
\sum_{k=0}^{L-1}\left(\frac{1}{\hat R_{th}P_I}\right)^k
 \sum_{\mathcal{U}} \prod_{i=1}^N \frac{1}{u_i !} {M^{(u_i)}_{W}\left(\frac{1}{\hat R_{th}P_I}\right)},
\end{align}
where $M^{(u_i)}_{W_i}(s)$ is computed from (\ref{gmgf}) by simply considering the relation $W_i = \frac{\gamma_i}{E_s/\overline N_0}$.
Thus, a closed-form outage probability expression for arbitrary $m$ is obtained in \eqref{CCI1}.
\begin{figure*}
\begin{align}\nonumber
&	\hat P_{\rm out} =
\\ & \nonumber
	\sum_{k=0}^{L-1}\left(\frac{1}{\hat R_{th}P_I}\right)^k
	\sum_{\mathcal{U}} \prod_{i=1}^N {{{m^m}{{\left( {1 + K + \frac{\overline W} {\hat R_{th}P_I}} \right)}^{m-u_i-1}}}}
	{{\overline W }^{u_i}}\sum\limits_{l = 0}^{u_i} {\binom{u_i}{l}\frac{{{{\left( m \right)}_l}}}{{l!}}} 
	\frac{{( {1 + K} )}^{l + 1}K}{{{{\left[ {m\left( {1 + K} \right) +\left( {m + K - K\Delta } \right)\overline W /{\hat R_{th}P_I}} \right]}^{m+l}}}}
\nonumber	\\&
\times	\sum\limits_{q = 0}^l 
	\binom{l}{q}
	{
		\left( {1 - \Delta } \right)^{l - q}}
		{\left( {2\Delta } \right)^q}
	\frac{{\Gamma \left( {\frac{1}{2} + q} \right)\Gamma \left( {\frac{1}{2}} \right)}}{{\pi\Gamma (q + 1)}} 
	{{\kern 1pt}_2}{F_1}\left( {{m+l},\frac{1}{2} + q;q + 1;\frac{-{2K\Delta \overline W }}{{m\left( {1 + K} \right){\hat R_{th}P_I} + \left( {m + K - K\Delta } \right)\overline W }}} \right).\label{CCI1}
		\\
	\hline \nonumber
\end{align}
\end{figure*}


\section{Numerical results}
In this section, numerical results obtained from our analytical derivations are presented for the new expressions of
the PDF and for the outage probability in the presence of interference in the two scenarios 
defined in the previous section. Results of scenario A are based on the connection between FTR and Nakagami-$m$ fading for  integer $m$, while
the connection between RS and FTR fading for arbitrary $m$ is exploited in scenario B.
These results are validated by Monte Carlo simulations, showing an excellent agreement.

In Fig. \ref{pdf}, the PDF of the FTR fading model is plotted for two values of the $m$ parameter: one integer ($m=3$) and one 
non-integer $m=1.5$ which has been obtain from \eqref{fNak} (for the case of $m$ integer) and \eqref{fFTR} (for both cases). Simulation 
results show a perfect match to the analytical results in all cases.

Regarding outage probability results, for scenario A we define the normalized average SINR as
\begin{align}
\overline{SINR}= \frac{\overline W}{L P_I +N_0},
\label{nsinr}
\end{align}
where $\overline W$ is the mean received power of the desired signal  and $P_I$ is the mean received power of an interferer.

\begin{figure}[t]
	\centering
	\includegraphics[width=1.05\linewidth]{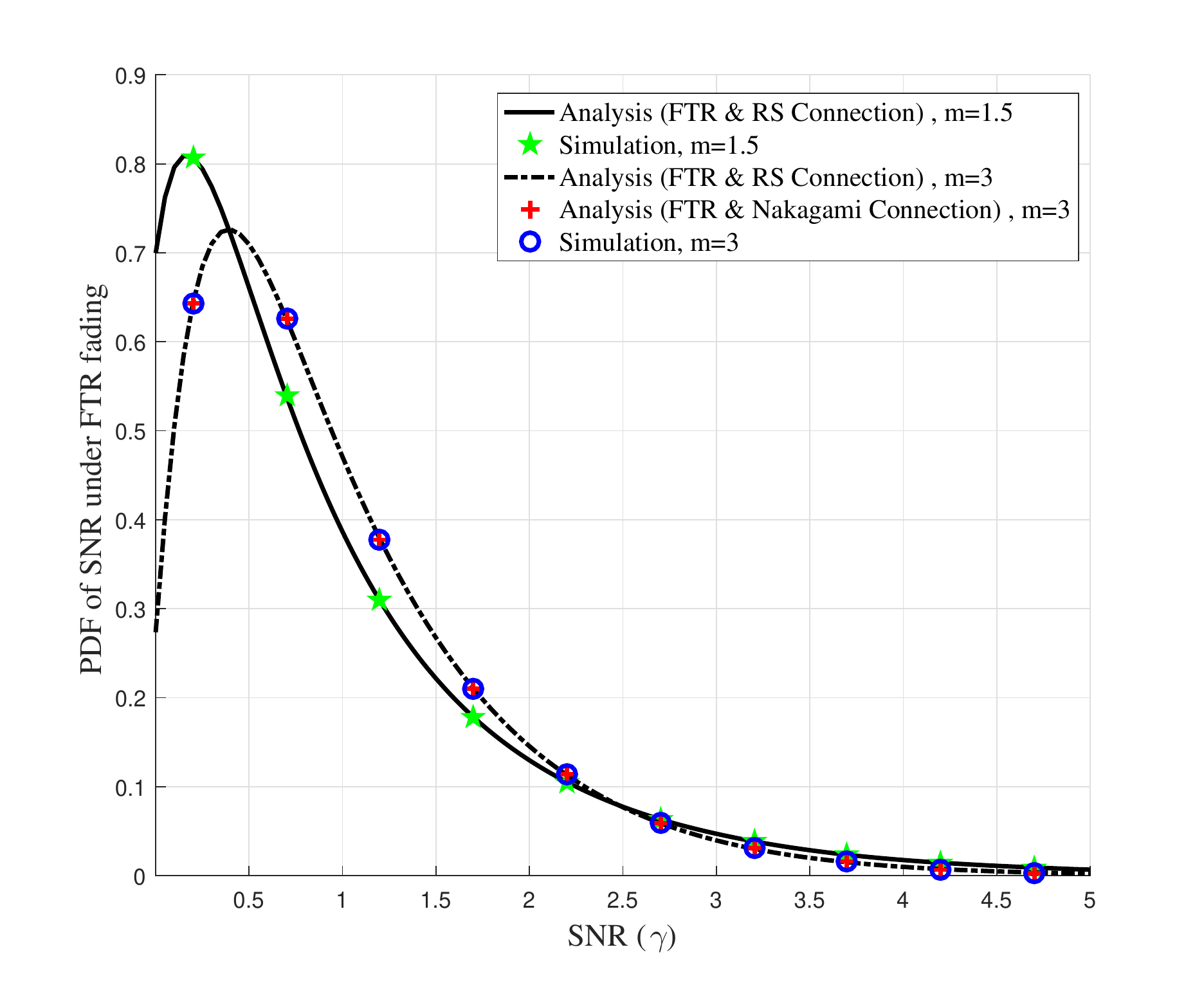}
	\caption{ PDF of the SNR ($\gamma$) under FTR fading obtained using the FTR-RS connection ($m=1.5,3$) and the
FTR-Nakagami-$m$ connection ($m=3$). Monte-Carlo simulation are also plotted. }
	\label{pdf}
\end{figure}
Figs. \ref{fig4} and \ref{fig5} present the outage probability of a system with CCI and background noise in terms of the 
normalized SINR (dB) as defined in (\ref{nsinr}).
In Fig. \ref{fig4}, the effect of different parameters of the FTR fading (including $m$, $K$, and $\Delta$) on the outage probability is evaluated.
It can be observed that, with similar $K$ and $\Delta$, the outage probability decreases as $m$ increases due to the reduction 
of the channel fluctuations. Also, for the given values of $m$ and $\Delta$, the outage probability declines by increasing $K$ 
from 10 to 15, which gives a measurement of the strength of the specular components to the diffuse component. 
Moreover, decreasing $\Delta$ from 0.6 to 0.2  for the same $m$ 
and $K$ provides a lower outage probability, since the reduction of $\Delta$ results in a lower similarity between the specular 
components in the FTR fading model, which therefore are less probable to cancel each other, as they have independent phases.
Fig. \ref{fig5} compares the effect of different values of the SINR threshold ($R_{th} = 6,8,10$). As expected, lower values of the SINR threshold 
yield lower values of the outage probability. Analytical results are verified by simulation.

\begin{figure}[t]
	\centering
	\includegraphics[width=1.05\linewidth]{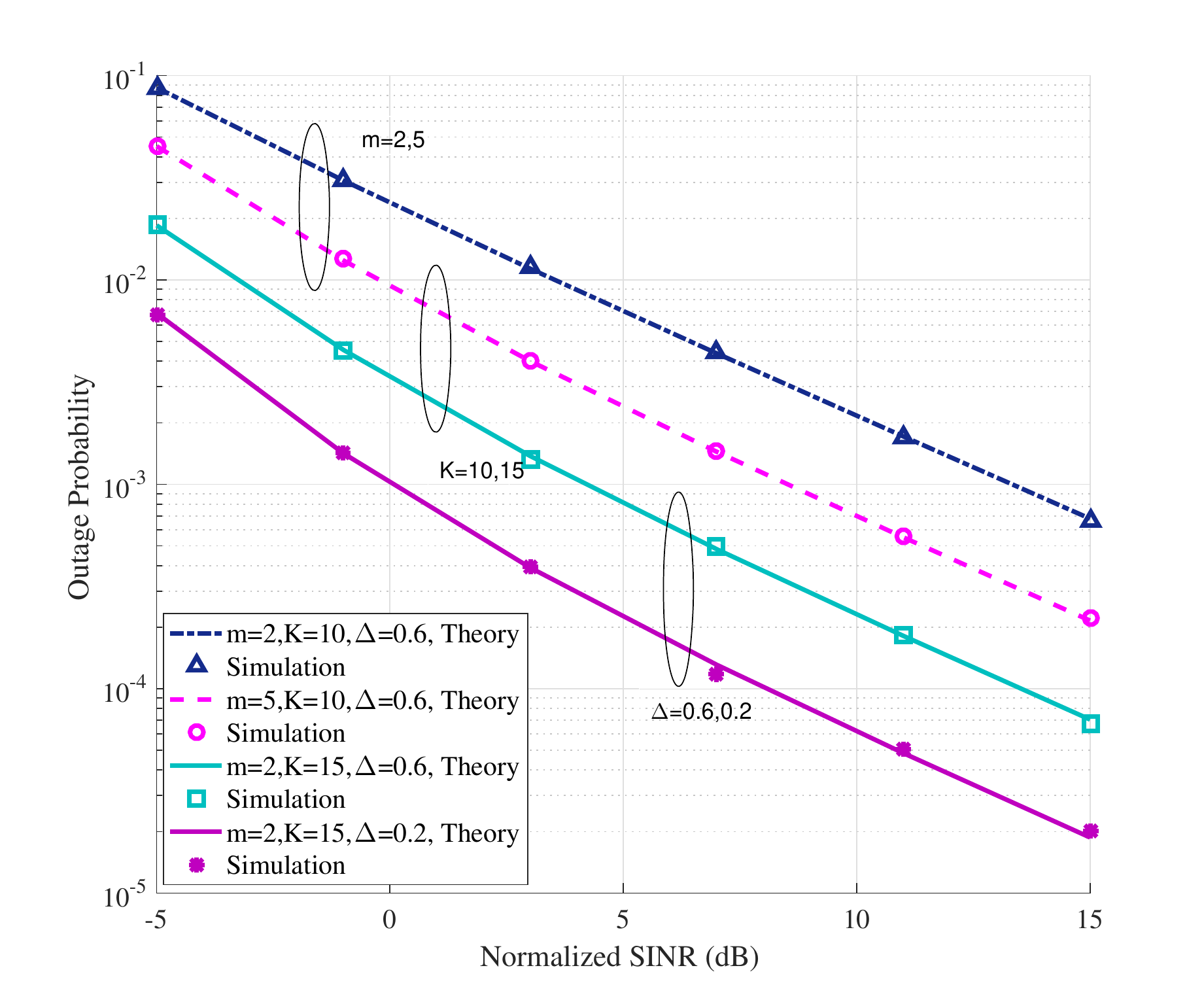}
	\caption{Analytical and simulation results for the outage probability in the presence of CCI and background noise versus 
normalized SINR (dB) for different $m$, $K$ and $\Delta$ with $L=2$, $P_I=0.01$  and $R_{th} =1$.}
	\label{fig4}
\end{figure}
\begin{figure}[t]
	\centering
	\includegraphics[width=1.05\linewidth]{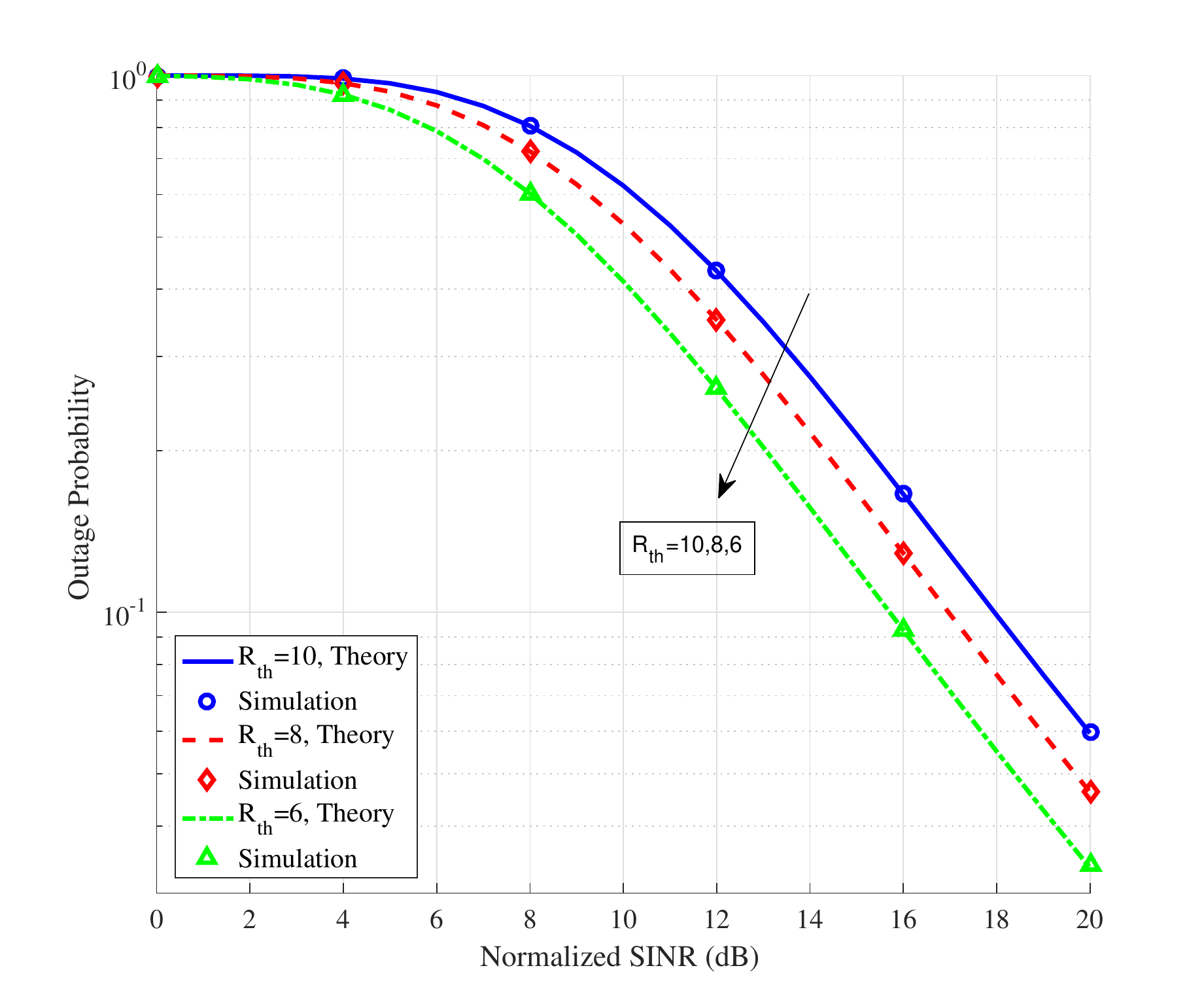}
	\caption{Analytical and simulation results for the outage probability in the presence of CCI and background noise versus 
normalized SINR (dB) for different values of the SINR threshold ($R_{th}$) with $m=2$, $\Delta = 0.6$ and $K=10$.}
	
	\label{fig5}
\end{figure}

Figs. \ref{fig2} and \ref{fig3} show the performance of the outage probability in a noise-limited multi-antenna system performing MRC.
Fig. \ref{fig2} provides the performance of the outage probability for different values of $m$ $(0.5,1,1.5,2.5)$. 
It can be observed that parameter $m$ has a significant impact on performance, in particular for lower values
of this parameter. For instance, for a given SIR threshold $\hat R_{th} = 0$ dB, the outage probability is
0.01 for $m=1$, while for $m=0.5$ the outage probability becomes 0.003. 
The analytical result, which are calculated using (\ref{CCI1}), are verified by Monte Carlo simulations.
Fig. \ref{fig3} shows  analytical results for the outage probability considering different numbers of receive antennas, 
$N=1,2,4$, and interferers, $L=1,2$. By exploiting the MRC technique at the multi-antenna receiver, 
it is shown that, as expected, increasing the number of antennas reduces the outage probability and improves the system performance in the presence of 
interference.

\begin{figure}[t]
	\centering
	\includegraphics[width=1.05\linewidth]{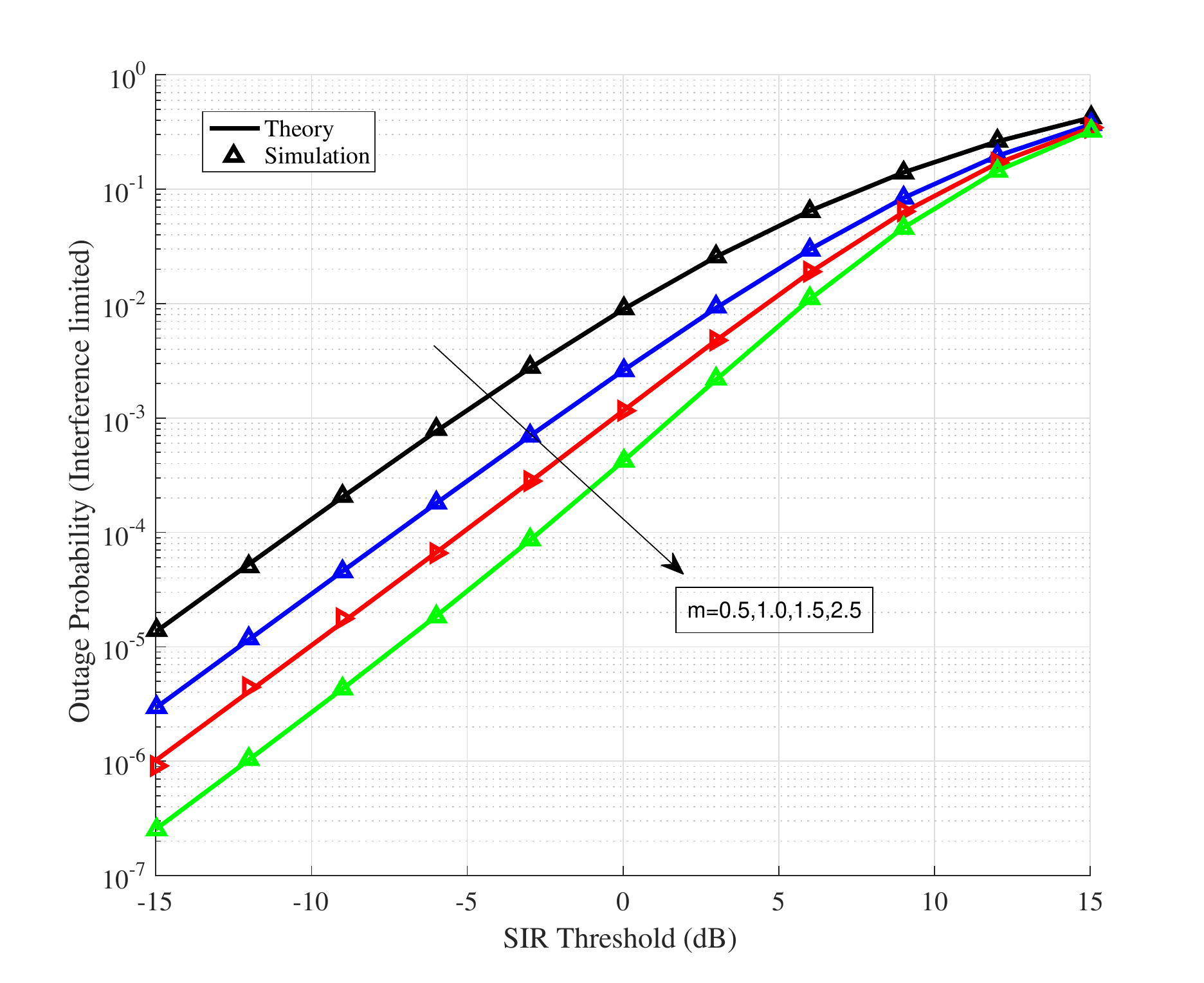}
	\caption{Analytical and simulation results for the outage probability versus SIR threshold (dB) in a interference-limited system for different values of $m$ $(0.5,1,1.5,2.5)$ with  $K=10$, $\Delta = 0.6$, $N = 2$, $P_I =1$ and $L= 1$.}
	\label{fig2}
\end{figure}
\begin{figure}[t]
	\centering
	\includegraphics[width=1.05\linewidth]{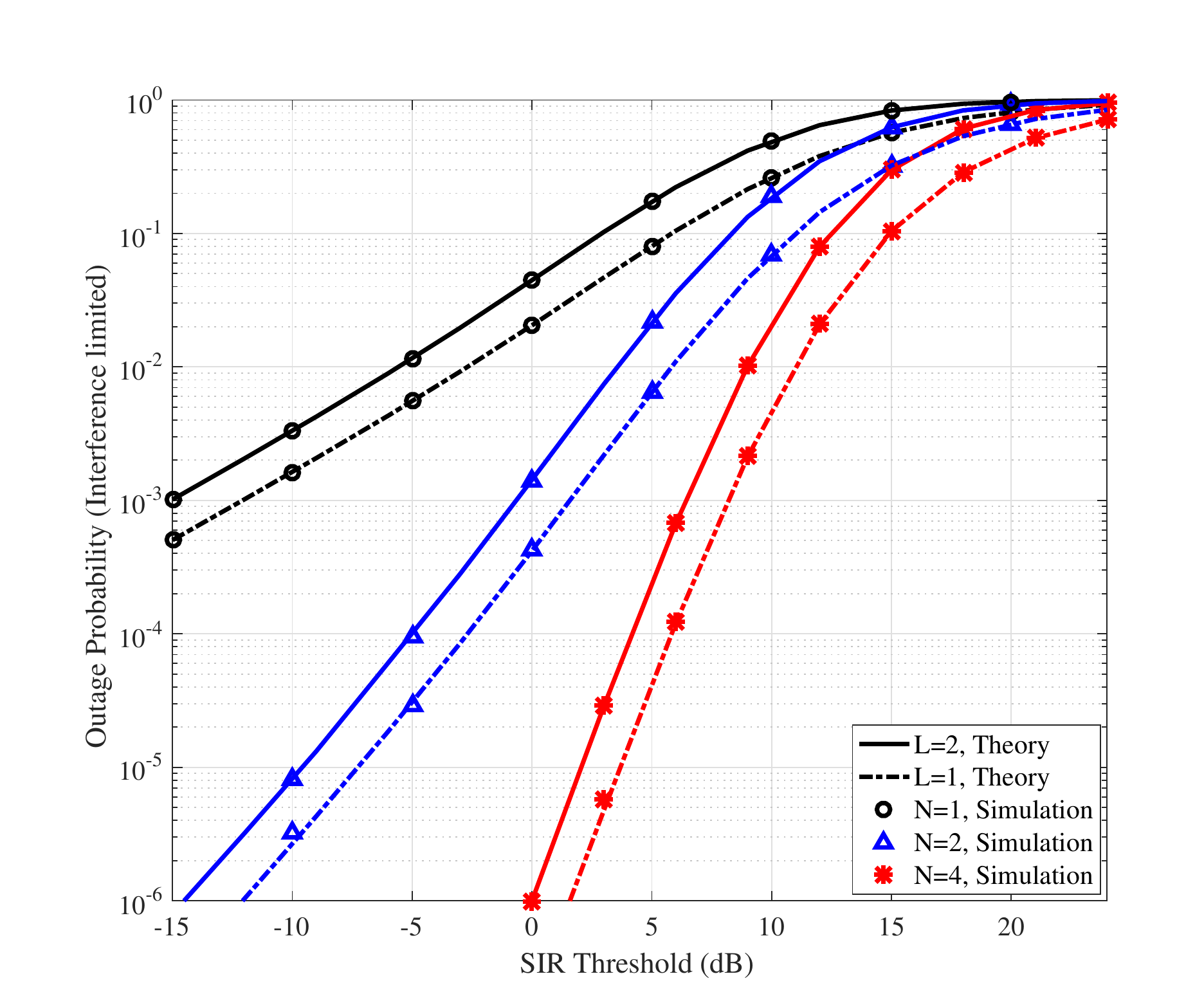}
	\caption{Analytical and simulation results for the outage probability versus SIR threshold (dB) for different numbers of antennas ($N$) and interferers ($L$) with $K=10$, $\Delta = 0.6$, $m = 2.5$, and $P_I=1$.}
	\label{fig3}
\end{figure}

\section{Conclusion}
We presented two flexible connections that describe the relationship between the fluctuating two-ray fading and two other 
fading models: ($i$) Rician Shadowed and ($ii$) Nakagami-$m$. In particular, we provided a relationship that can be leveraged 
to derive performance metrics under FTR fading by using the available performance results under RS fading for any arbitrary $m$, as well 
as utilizing the existing performance results under Nakagami-$m$ fading model for integer $m$. Based on these novel formulations, we provided 
new analytical results for very relevant Laplace-domain statistics which were not previously available for the case of FTR 
fading, finding direct application in a number of scenarios of interest in communication theory.}

%
%


\section*{Appendix 
	\\
	 Solving Integral I1 and I2}
 The integral I1 for an arbitrary non-negative real number $\nu$ and $a > \left| b \right|$ is defined as
 		\begin{align}
 \text{	I1} \triangleq	\frac{{\rm{1}}}{\pi }\int_0^\pi  {\frac{{dx}}{{{{\left( {a + b\cos x} \right)}^{v + 1}}}}} .
  	\end{align}
 	Let us consider the following equality by changing variable $u = c\,t $, as:
 	 		\begin{align}\nonumber
 &	\int_0^\infty  {{t^v}{e^{ - c \cdot t}}dt} 
  	 = \frac{1}{{{c^{v + 1}}}}\int_0^\infty  {{u^v}{e^{ - u}}du}  = \frac{1}{{{c^{v + 1}}}}\Gamma \left( {v + 1} \right).
  	\end{align}
  Therefore,
\begin{align}
 \text{I1} = \frac{{\rm{1}}}{\pi }\int_0^\pi  {\left[ {\frac{1}{{\Gamma \left( {v + 1} \right)}}\int_0^\infty  {{t^v}{e^{ - \left( {a + b\cos x} \right) \cdot t}}dt} } \right]} \,dx.
 	 \end{align}
 	Now, by interchanging the order of integration we can write
 	\begin{align}
 &	\text{I1} = \frac{1}{{\Gamma \left( {v + 1} \right)}}\frac{{\rm{1}}}{\pi }\int_0^\infty  {\left[ {\int_0^\pi  {{t^v}{e^{ - \left( {a + b\cos x} \right)  t}}dx} } \right]} \,dt
 	\\&
 	 = \frac{1}{{\Gamma \left( {v + 1} \right)}}\int_0^\infty  {{t^v}{e^{ - a  t}}\left[ {\frac{{\rm{1}}}{\pi }\int_0^\pi  {{e^{ - b  t\cos x}}dx} } \right]} \,dt\\
 &	= \frac{1}{{\Gamma \left( {v + 1} \right)}}\int_0^\infty  {{t^v}{e^{ - at}}{I_0}\left( {bt} \right)} dt \label{bessel}
 \\&
 = \frac{1}{{{{\left( {\sqrt {{a^2} - {b^2}} } \right)}^{v + 1}}}}{P_v}\left( {\frac{a}{{\sqrt {{a^2} - {b^2}} }}} \right),
 \end{align}
where $P_v$ is the Legendre function and where we have used \cite[p. 196 (8)]{erdel} 
together with the fact that the Bessel function $I_0$ in \eqref{bessel} is an even function.
 
  The integral I2 for positive integer $P1$, arbitrary positive $P2$, arbitrary $\alpha$ and $|\beta|<1$ can be computed as
    \begin{align}
&  \text{I2} \triangleq \int_0^\pi  \frac{\left( {1 + \alpha \cos (\theta )} \right)^{P1}}{\left( {1 + \beta \cos (\theta )} \right)^{P2}}d \theta
  	\\ \nonumber
  &\overset{(A)}{=}
  \frac{1}{(1-\beta)^{P2}}\int_0^1 
  (1-\alpha+{2\alpha}x)^{P1} 	(1+\frac{2\beta}{1-\beta}x)^{-P2}
  \\& \times 
  (-1)x^{-\frac{1}{2}}(1-x)^{-\frac{1}{2}}
  dx
  \\ \nonumber
  &\overset{(B)}{=}
  \frac{-1}{(1-\beta)^{P2}}
  \sum_{q=0}^{P1} \binom{P1}{q} ({2\alpha})^q (1-\alpha)^{P1-q}
  \\& \times
  \int_0^1  
  x^{q-\frac{1}{2}}(1-x)^{-\frac{1}{2}} 	(1+\frac{2\beta}{1-\beta}x)^{-P2}
  dx
   \\ \nonumber
  &\overset{(C)}{=}
  \frac{-1}{(1-\beta)^{P2}}
  \sum_{q=0}^{P1} \binom{P1}{q} 
  ({2\alpha})^q (1-\alpha)^{P1-q}
  \frac{\sqrt{\pi}\Gamma\left(q+\frac{1}{2}\right)}{\Gamma(q+1)}
  \\
  & \times 
  {{\kern 1pt}_2}{F_1}\left(P2,q+\frac{1}{2} ;q + 1;-\frac{2\beta}{1-\beta} \right),
\end{align}
where $(A)$ followed from the change of variables $\cos (\theta )=2x-1$, yielding $d\theta = \frac{-dx}{\sqrt{x}\sqrt{1-x}}$, $(B)$ 
followed by using the binomial theorem $(a+b)^n=\sum_{q=0}^n \binom{n}{q}a^{n-q}b^q$ for positive integer $P1$, and $(C)$ is 
obtained from the integral representation of the Gauss hypergeometric function \cite[eq. (9.111)]{Gradsh}
\begin{align}
&	_2{F_1}\left( {a,b;c;z} \right)  \nonumber\\  
&= \frac{{\Gamma (c)}}{{\Gamma (b)\Gamma (c - b)}}\int_0^1 {{t^{b - 1}}{{\left( {1 - t} \right)}^{c - b - 1}}{{\left( {1 - tz} 
\right)}^{ - a}}dt}  \label{hyper}
\end{align}
where $a=P2$, $b=q+\frac{1}{2}$, $c=q+1$ and $z=-\frac{2\beta}{1-\beta}$. 



\end{document}